\documentclass{article}

\usepackage{amsmath,amssymb,upgreek,amsthm,epsfig,enumerate}

\usepackage{amssymb}

\newcommand{\citep}[1]{\cite{#1}}
\newcommand{\citet}[1]{\cite{#1}}

\newtheorem{definition}{Definition}[section]
\newtheorem{lemma}[definition]{Lemma}
\newtheorem{theorem}[definition]{Theorem}

\newtheorem{corollary}[definition]{Corollary}
\newtheorem{remark}[definition]{Remark}
\newtheorem{example}[definition]{Example}
\newtheorem{examples}[definition]{Examples}

\newcommand{\ii}{\mathrm{i}}
\newcommand{\e}{\mathrm{e}}
\newcommand{\dd}{\mathrm{d}}
\newcommand{\D}{\mathrm{D}}
\newcommand{\const}{\mathrm{const}}

\newcommand{\re}{\mathrm{Re}\;}
\newcommand{\im}{\mathrm{Im}\;}
\newcommand{\sgn}{\mathrm{sgn}}

\newcommand{\oo}{\mathrm{o}}

\newcommand{\upi}{\uppi}

\newcommand{\LICM}{\mathfrak{L}}
\newcommand{\BF}{\mathfrak{B}}
\newcommand{\CBF}{\mathfrak{F}}

\renewcommand{\Sigma}{\mathfrak{S}}
\newcommand{\CrF}{\mathfrak{Q}}

\title{On wave propagation in viscoelastic media with concave creep compliance}

\author{Andrzej Hanyga\\
ul. Bitwy Warszawskiej 1920 r. 14/52\\
02-366 Warszawa, Poland}

\usepackage{pifont}

\begin{document}

\maketitle

\begin{abstract}
It is proved that the attenuation function of a viscoelastic medium with a non-decreasing and concave 
creep compliance is sublinear in the high frequency range. 
\end{abstract}


\section*{Notation.}

\begin{tabular}{l l l}
$\mathbb{R}$ & the set of real numbers & \\
$\mathbb{C}$ & the set of complex numbers & \\
$f \circ g$ & composition & $ (f\circ g)(x) := f(g(x))$ \\
$\mathcal{L}[f](p) = \tilde{f}(p)$ & Laplace transform of $f(t)$ & $\tilde{f}(p) = \int_0^\infty \e^{-p t}\, f(t)\, \dd t$\\
$\hat{f}(k)$ & Fourier transform of $f(x)$ & $\hat{f}(k) = \int_{-\infty}^\infty \e^{-\ii k \cdot x} f(x)\, \dd x$
\end{tabular}

\section{Introduction.}

Ultrasonic tests of mechanical properties of polymers and bio-tissues assume that the material is viscoelastic. 
In applications a linear viscoelastic material is always assumed to have a completely monotonic (CM)
relaxation modulus \citep{Bland:VE,Day70b,BerisEdwards93}. This assumption has important consequences 
for the attenuation and dispersion of acoustic waves. In particular, the increase of attenuation 
with frequency in the high frequency range is slower than linear \citep{HanWM2013}. This fact cannot 
be verified experimentally because experimental test cover the range of frequencies below 250 MHz. 
The results of ultrasonic tests of polymers and bio-tissues suggest a power-law dependence of 
attenuation in the frequency range accessible to measurements, $\mathcal{A}(\omega) = A\, \omega^\alpha$ with 
$\alpha > 1$. This power-law dependence is then extrapolated to high frequency range, where it contradicts 
various general principles. It is then explained in terms of {\em ad hoc} lossy wave equations 
\cite{Szabo1,Szabo2,KellyMcGoughMeerschaert08} and wave equations with fractional Laplacians 
\cite{ChenHolm04,TreebyCox10} which are incompatible with viscoelastic constitutive equations.

Superlinear growth of attenuation in the high frequency range is inconsistent with finite speed of propagation. 
The requirement of finite speed of propagation imposes an upper growth limit $\omega/\ln(\omega)$
on high-frequency attenuation: the propagation speed is finite if  
$\omega/\ln^\alpha(\omega)$, $\alpha > 1$ but not for $\alpha \leq 1$. Superlinear growth of attenuation 
in the high frequency range is also inconsistent with viscoelastic
constitutive equations even if the propagation speed is infinite, for example it is inconsistent with Newtonian viscosity,
where the attenuation satisfies the power law $\const \times \omega^{1/2}$,
and with strongly singular locally integrable completely monotonic relaxation moduli \cite{HanWM2013}.

According to the theory of wave propagation in viscoelastic media with completely monotonic relaxation moduli, developed in 
\citet{SerHan2010,HanWM2013}, in the low frequency range the wave attenuation in a viscoelastic solid  
increases like $\omega^{1 + \alpha}$, $0 < \alpha \leq 1$,  while a 
viscoelastic fluid is characterized by a rate of increase $\omega^\alpha$,  $0 < \alpha < 1$. Low frequency
behavior of attenuation is accessible to measurements and is consistent with observations. In the high frequency 
range the growth of attenuation is however sublinear. High frequency attenuation is also sublinear 
in viscoelastic media with spatial derivative operators of fractional order \cite{SerHanJMP}.

Although practical modeling of viscoelastic media routinely assumes that the relaxation modulus
is completely monotonic, there is hardly any a priori
evidence for this assumption. The only general physical principle -- the fluctuation-dissipation theorem -- 
implies that the relaxation modulus 
is completely positive \citep{HanSerDielectrics07}. The consequences of complete positivity for the 
dependence of attenuation on frequency are however at present unknown. In this paper we explore the consequences of 
a different extension of the complete monotonicity assumption, based on the
assumption that the creep compliance is non-decreasing and concave. Non-negative 
functions with these two properties are said to belong to the CrF class, which is an abbreviation of 
Pr\"{u}ss' unfortunate term "creep functions". The above assumption, proposed 
and explored in Jan Pr\"{u}ss' book \citep{Pruss}, is much weaker than the hypothesis of a completely 
monotonic relaxation modulus, assumed in \citet{SerHan2010,HanWM2013}. The latter is equivalent to 
the creep compliance being a Bernstein function, which is defined in terms of an 
infinite sequence of inequalities for the derivatives of the creep compliance of arbitrary order. 
In the CrF class only the three first of these inequalities are required to hold.
The Bernstein class and the CrF class includes Newtonian viscosity characterized by linear creep.
Unlike the Bernstein  property, the CrF property of a function can be easily visually recognized from its plot. 
Furthermore, the CrF property is stable with respect to perturbations which are small in the 
sup norm. 

It will be shown here that wave attenuation in media with CrF creep compliances is bounded by a function 
$a + b \, \vert \omega \vert$. A sharper form of this bound was obtained for creep compliances in the Bernstein class 
in \citet{SerHan2010,HanWM2013}. 
The theory developed in \citet{SerHan2010,HanWM2013} relies on a spectral analysis of attenuation and dispersion.
Unfortunately a spectral analysis of attenuation and dispersion similar to the completely monotonic case
has not been possible but it is very likely that the bound obtained here for the CrF class of
creep compliances can be significantly sharpened.

After a recapitulation of necessary facts from the theory of CM functions and Bernstein functions 
(Sec.~\ref{sec:recoll}) the essentials of Pr\"{u}ss' theory are recalled in Sec.~\ref{sec:Pruss}. 
The main result concerning the bound on wave attenuation is derived in Sec.~\ref{sec:proof}.

\section{Some properties of completely monotonic and Bernstein functions.}
\label{sec:recoll}

\begin{definition}
A real function $f$ on $]0,\infty[$ is said to be completely monotonic (CM) if $f \in \mathcal{C}^n$
and $(-1)^n \, \D^n\, f(x) \geq 0$ for for every integer $n \geq 0$.\\
\end{definition}

\begin{theorem} (Bernstein) \\
A real function $f$ on $ ]0,\infty[$ is CM if and only if there is a positive Radon measure 
$\mu$ on $[0,\infty[$ such that
\begin{equation} \label{eq:Bernstein}
f(t) = \int_{[0,\infty[} \e^{-t r} \mu(\dd r), \qquad t > 0
\end{equation}
\end{theorem}

A CM function is locally integrable if and only it is integrable over the interval $[0,1]$.
The set of locally integrable CM (in short LICM) functions is denoted by $\LICM$.

\begin{theorem} \citep{HanDuality,HanWM2013}\\
A real function $f$ is LICM if and only if there is a positive Radon measure $\mu$ such that 
\begin{equation} \label{eq:doss}
\int_{[0,\infty[} \frac{\mu(\dd r)}{1 + r} < \infty
\end{equation}
and eq.~\eqref{eq:Bernstein} holds.
\end{theorem}

\begin{definition}
A real function $f$ defined on $[0,\infty[$ is said to be a Bernstein function (BF) if it has a derivative
which is a LICM function.\\
The set of Bernstein functions is denoted by the symbol $\BF$.
\end{definition}

It is obvious that a Bernstein function is non-negative, non-decreasing and concave.

From Bernstein's theorem the following representation of Bernstein functions 
can be derived:
\begin{theorem}\label{th:BernsteinSpectral} \citep{BergForst}\\
A function $f: \mathbb{R}_+ \rightarrow \mathbb{R}$ is a Bernstein function if 
and only if the are two non-negative real numbers 
$a, b$ and a positive Radon measure $\lambda$ on $\mathbb{R}_+$ such that 
\begin{equation} \label{eq:BS1}
\int_{]0,\infty[} \frac{r}{1 + r} \, \lambda(\dd r) < \infty
\end{equation}
and
\begin{equation} \label{eq:BernsteinSpectral}
f(t) = a + b \, t + \int_{]0,\infty[} \left[1 - \e^{-r t}\right] \, \lambda(\dd r), 
\qquad t > 0
\end{equation}
\end{theorem}

Inequality~\eqref{eq:BS1} is equivalent to the following pair of inequalities
\begin{equation} \label{eq:twoineq}
\int_{]0,1]} r \nu(\dd r) < \infty, \qquad \int_{]1,\infty[} \nu(\dd r) < \infty
\end{equation}

\begin{theorem} \label{thm:lim}
If $f$ satisfies eq.~\eqref{eq:BernsteinSpectral} with $\lambda$ satisfying
eq.~\eqref{eq:BS1}, then 
$\lim_{t \rightarrow \infty} f(t)/t = b$. 
\end{theorem}
\begin{proof}
Expansion in power series shows that 
$(1 + t)^2/2 \leq \e^t \leq (1 + t)^2$ for all $t \geq 0$.  Hence
$(1 + t)^{-2} \leq \e^{-t} \leq 2 (1 + t)^{-2}$ for $t \geq 0$ and, upon integration,
$t/(1 + t) \leq 1 - \e^{-t} \leq 2 t/(1 + t)$. 

Now
\begin{equation} \label{eq:xxx}
\frac{f(t)}{t} = \frac{a}{t} + b + \int_{]0,\infty[} \frac{1 - \e^{-r t}}{t} \, \lambda(\dd r)
\end{equation}
The integrand of the integral on the right-hand side of eq.~\eqref{eq:xxx}
tends to 0 as $t \rightarrow \infty$. The integrand is also bounded from above by $2 r/(1 + r t)$
and by $2 r/(1 + r)$ for $t \geq 1$. In view of eq.~\eqref{eq:BS1} and the Lebesgue Dominated Convergence
Theorem the integral on the right-hand side of eq.~\eqref{eq:xxx} tends to 0 as $t \rightarrow \infty$.
The thesis follows from eq.~\eqref{eq:xxx}. 
\end{proof}

\begin{theorem} \label{thm:BFCM} \citep{BernsteinFunctions}\\
If the real function $f$ defined on $]0,\infty[$ is CM and $g \in \BF$,
then the composition $f\circ g$ is CM.\\
If $f, g$ are Bernstein functions on $[0,\infty[$, then the composition 
$f \circ g$ is a Bernstein function.\\
If a positive real function $f$ defined on $[0,\infty[$ is a BF then
$1/f$ is CM on $]0,\infty[$.
\end{theorem}

\section{The CrF class of functions.}
\label{sec:Pruss}

In every neighborhood of a bounded CM function in the space of bounded continuous functions 
$\mathcal{C}^0([0,\infty[)$ endowed with
the sup norm there is a function $g$ which is not CM, for example $g(x) = f(x) + \sin(x)/n$. 
Hence the CM property of the relaxation modulus cannot be experimentally confirmed by a direct 
method although it can be inconsistent with experimental data in a specific case. 
Thus complete monotonicity has to be considered as an a priori constraint on
the interpolation of experimental data by matching them against a multi-parameter family
of completely monotonic functions (Prony sums, Cole-Cole or Havriliak-Negami functions, the
Kohlrausch-Watts-Williams function). 

It is therefore interesting to note that the linear upper bound on the 
high frequency behavior of the attenuation function is not a consequence of complete
monotonicity of relaxation moduli. 
As an alternative to viscoelasticity based on completely monotonic relaxation functions,
a theory based on an 
experimentally verifiable concavity property of creep compliances will be presented here 
following \citet{Pruss}.

\begin{definition}
Let $I$ be a segment of the real line.

A function $f: I \rightarrow \mathbb{R}$ is said to be \emph{concave} if
$$\vartheta\,f(x) + (1-\vartheta)\, f(y) \leq f(\vartheta\, x + (1-\vartheta)\,y)$$
for all $x, y \in I$, $0 \leq \vartheta \leq 1$.
\end{definition}

Bernstein functions are non-negative, non-decreasing and concave. 

\begin{definition}
A function $f: \mathbb{R}_+\rightarrow \mathbb{R}$ is said to be a {\em CrF} if
it is non-negative, non-decreasing and concave.

$\CrF$ denotes the set of all the CrFs. 
\end{definition}

Every Bernstein function is obviously a CrF.

\begin{example}
The function 
\begin{equation} \label{eq:CrFexa}
f(x) := \int_0^x \e^{-y^\alpha} \, \dd y
\end{equation}
(Fig.~\ref{fig:CrF})
is a CrF for $\alpha > 0$, but it is not a BF if $\alpha > 1$.
Indeed, if $\alpha > 0$ then $f, f^\prime > 0$ while $f^{\prime\prime} < 0$.

It $\alpha \leq 1$ then $f^\prime(x) \equiv \e^{-x^\alpha}$ is CM, hence $f \in \BF$. 

If $\alpha > 1$ then $f^{\prime\prime\prime}(x) \equiv \alpha \, x^{\alpha-2}\, [\alpha \, x 
- (\alpha - 1)] \, \e^{-x^\alpha}$ changes sign at $x = (\alpha-1)/\alpha$. Hence
$f$ is not a Bernstein function in this case.
\end{example}

\begin{figure}
\begin{minipage}{0.48\textwidth}
\epsfig{file=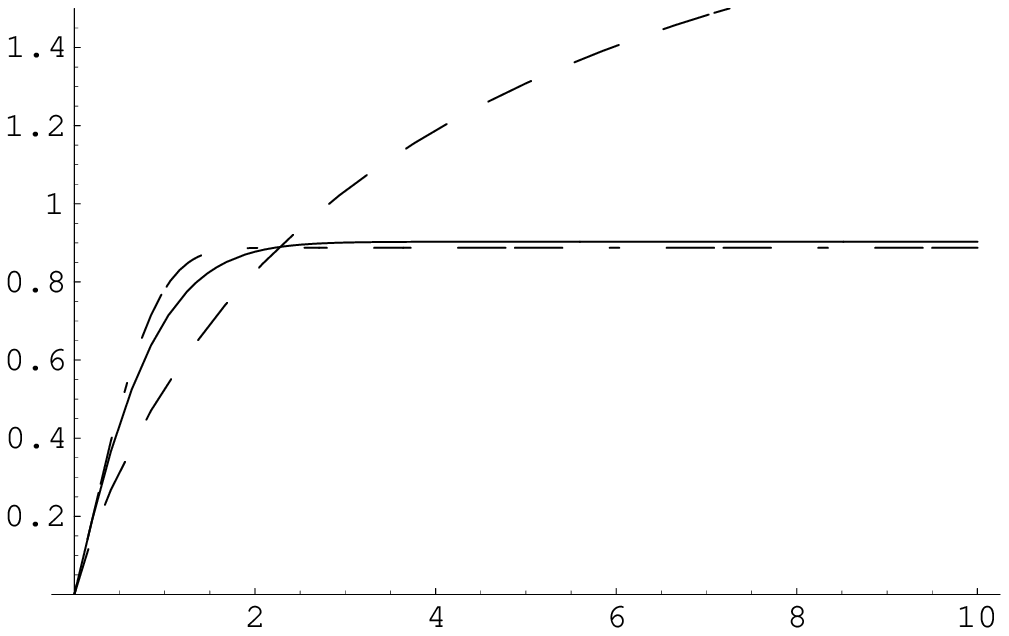,width=\textwidth}
\begin{center}  $(a)$
\end{center}
\end{minipage}
\begin{minipage}{0.48\textwidth}
\epsfig{file=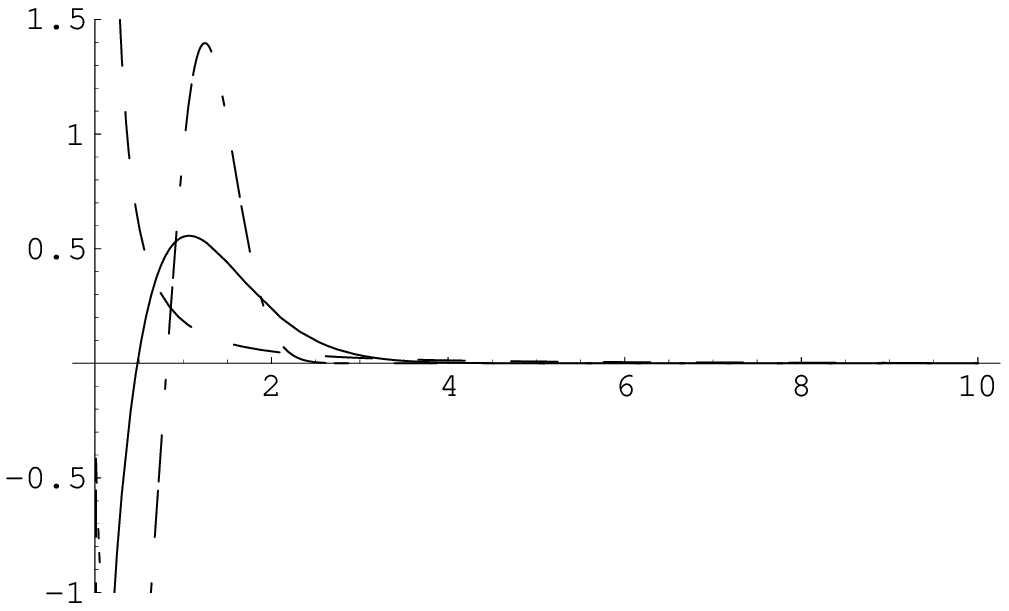,width=\textwidth}
\begin{center}  $(b)$
\end{center}
\end{minipage}
\caption{$(a)$: The function defined by \eqref{eq:CrFexa}, for $\alpha=1/2$ (dashed),
$\alpha=3/2$ (solid) and $\alpha=5/2$ (dot-dashed). $(b)$ The 3rd order derivative
$\D^3 f$ for the same values of $\alpha$.}
\label{fig:CrF}
\end{figure}

\begin{theorem} \label{CFrep}
Every CrF $f$ can be expressed in the form
\begin{equation} \label{eq:CrFrep}
f(t) = a + b \, t + \int_0^t k(s) \, \dd s
\end{equation}
where $a, b$ are non-negative numbers and the locally integrable function 
$k$ is non-negative, non-increasing and 
$k(s) \rightarrow 0$ for $s \rightarrow \infty$. 
\end{theorem}
\begin{proof}
Let $a := f(0)$ and 
$b := \inf_{t>0} f(t)/t$. The function $b$ assumes finite values because concavity 
$$f(\vartheta \, t) \geq (1-\vartheta) \, a + \vartheta \, f(t)$$
implies that
$$\frac{f(\vartheta \, t)}{\vartheta \, t} - 
\frac{1-\vartheta}{\vartheta \, t} a \geq \frac{f(t)}{t}$$
for $t > 0$ and $0 < \vartheta < 1$,
hence $f(t)/t$ is decreasing and has a finite limit $b$.
The function $f$ is non-negative and non-decreasing, hence it is bounded on 
an interval $[0,\varepsilon]$, $\varepsilon > 0$. 
Since $-f$ is convex, $f$ has a derivative $f^\prime$ at all but a countable 
set of points. The derivative $f^\prime$ is almost everywhere non-negative (because $f$ is non-decreasing)
and almost everywhere non-increasing (because $f$ is concave). The second-order derivative
$f^{\prime\prime}$ exists and is non-positive almost everywhere. Hence
$f^\prime(t) \leq f^\prime(t) + t \, f^{\prime\prime}(t) \equiv \dd [t \, f^\prime(t)]$
and thus,upon integration and dividing by $t$, $f^\prime(t) \geq f(t)/t - a/t$. Since $f(t)$ is 
non-decreasing, the right-hand side of this inequality is $\geq b$.   

The function $k(t) := f^\prime(t) - b$
is measurable, almost everywhere non-negative,  and non-increasing. 
Hence $k$ is locally integrable near 0 (otherwise $f$ would be infinite at 0). 
In view of monotonicity $\lim_{t \rightarrow \infty} k(t)$ exists and, in view 
of the definition of $b$, it vanishes.
\end{proof}  \hfill $\Box$

If $f$ is a CrF function satisfying eq.~\eqref{eq:CrFrep} then $\D f$ exists almost everywhere and $\D f = k$,
$f(0) = a$ and $\lim_{t\rightarrow \infty} \D f(t) = b$. 

\begin{examples}\label{ex:CrF} \mbox{ }
\begin{enumerate}
\item The function $f(t) := a + b\, t + c \, t^\alpha$, 
$a,b,c \geq 0$, $0 < \alpha < 1$,
is a CrF function, with $k(t) = \alpha\, t^{\alpha-1}$. 
\item The function $f$ defined above is not a CrF if $\alpha > 1$. 
\end{enumerate}
\end{examples}

\begin{theorem} \label{thm:BCrF} \mbox{ }
\begin{enumerate}[(i)]
\item If $f \in \CrF$ and $g(x) := x^2\, \tilde{f}(x)$ then $g \in \BF$.
\item If $f \in \BF$ then there is a CrF $g$ such that $f(x) = x^2\, \tilde{g}(x)$.
\end{enumerate}
\end{theorem}
\begin{proof}
Ad (i)\\
Eq.~\eqref{eq:CrFrep} implies that $f$ has the Laplace transform
$$\tilde{f}(p) = \frac{a}{p} + \frac{b}{p^2} + \frac{1}{p} \tilde{k}(p)$$
hence
$$p^2 \, \tilde{f}(p) = a\, p + b + p\, \tilde{k}(p)$$
But
$$
p\, \tilde{k}(p) = p \int_0^\infty \e^{-p r} \, k(r)\, \dd r = 
\int_0^\infty \dd \left[1 - \e^{-p r}\right] \, k(r) = 
-\int_0^\infty \left[ 1 - \e^{-r p}\right] \, \dd k(r)
$$
Define the Radon measure $\lambda$ on $\mathbb{R}_+$ by the formula $\int \varphi(r) \, \lambda(\dd r) =
-\int \varphi(r) \, \dd k(r)$ for continuous functions $\varphi$ with compact support in $\mathbb{R}_+$. 
The measure $\lambda$ is positive and 
$$\int_{]0,\infty[} \frac{r}{1 + r} \lambda(\dd r) = - \int_{]0,\infty[} 
\frac{r \, \dd k(r)}{1 + r} = \int_0^\infty \frac{k(r)}{(1 + r)^2} \dd r < \infty$$
We have used the facts that $\lim_{r\rightarrow 0+} r\, k(r) = 0$ (because $k$ is locally integrable) 
and $k$ is non-increasing. Consequently the Radon measure $\lambda$
satisfies inequality \eqref{eq:BS1}. We now have the identity 
$$p\, \tilde{k}(p) = \int_0^\infty \left[ 1 - \e^{-r p}\right] \, \dd k(r)$$
We thus see that the function $g$ has the integral representation \eqref{eq:BernsteinSpectral} and is thus a 
Bernstein function, q.e.d.

If the limit $k_0 := \lim_{r\rightarrow 0+} k(r)$ exists then $\nu$ can be defined as the Borel measure 
$\nu(]r,s]) = k(r) - k(s)$ for $0 \leq r \leq s$, provided 
$k$ is redefined as a right-continuous function. 

Ad (ii)\\
The function $f$ satisfies eq.~\eqref{eq:BernsteinSpectral} with the Radon measure $\lambda$ 
satisfying inequality \eqref{eq:BS1}. Since $r/(1 + r) \geq 1/2$ for $r > 1$, it follows from eq.~\eqref{eq:BS1} 
that for every $r > 0$ the right-continuous function 
$k(r) := \lambda(]r,\infty[) = \int_{]r,\infty[} \lambda(\dd \xi)$ is non-negative and finite. Since
$$\int_a^b k(r) \, \dd r \equiv  \int_a^b r \,\lambda(\dd r) + b\, k(b) - a \, k(a)$$
is in view of eq.~\eqref{eq:twoineq}$_1$ finite for $0 \leq a < b$, 
the function $k$ is locally integrable on $\overline{\mathbb{R}_+}$. It is also non-negative, 
non-increasing and, on account of \eqref{eq:BS1}, it tends to zero at infinity. The integral 
in eq.~\eqref{eq:BernsteinSpectral} can be integrated by parts:
$$-\int_{]0,\infty[} \left(1 - \e^{-p r}\right)\, \dd k(r) =
p \int_{]0,\infty[} \e^{-p r} \, k(r) \, \dd r = p \, \tilde{k}(p)$$
Hence $f(t) = p^2 \, \tilde{g}(p)$ with $g(t) := a\, t + b + \int_0^t k(s)\, \dd s$
and $g \in \CrF$ by Theorem~\ref{CFrep}.
\end{proof}

\begin{remark}
The definition of a complete Bernstein function (Appendix~\ref{app:A}) is equivalent to
the statement that, assuming again the equation $g(x) = x^2 \, \tilde{f}(x)$, 
$g \in \CBF$ if and only if $f \in \BF$.
 \end{remark}

\section{The wave number function in materials with a CrF creep compliance.}
\label{sec:waves}

We now consider a viscoelastic material whose creep compliance $J$ is a CrF.

Consider the Fourier and Laplace transform of the viscoelastic equation of motion assume
\begin{equation} \label{eq:dispeq}
\rho\, p^2 \, \hat{\tilde{u}} = -p\, \tilde{G}(p)\,k^2 \, \hat{\tilde{u}}
\end{equation}
and note that $p\,\tilde{G}(p) = 1/[p \, \tilde{J}(p)]$ \citep{HanDuality}. Define the wave number
function $\kappa(p)$, $p \in \mathbb{C}$,  in such a way that $k = -\ii \kappa(p)$ and $\re \kappa(p) \geq 0$
is a solution of eq.~\eqref{eq:dispeq}.

In view of eq.~\eqref{eq:dispeq} the wave number function can be expressed in terms 
of the Laplace transform $\tilde{J}$ of the creep compliance  
\begin{equation}
\kappa(p) = \rho^{1/2}\, p^{1/2}\, \left[p^2\, \tilde{J}(p)\right]^{1/2}
\end{equation}
If $J \in \CrF$ then, by Theorem~\ref{thm:BCrF}, $f(p) := p^2\, \tilde{J}(p)$ is a Bernstein function.
Hence $\kappa = [p f(p)]^{1/2} \in \mathcal{K}_\CrF := p^{1/2}\, \BF^{1/2}$, where $\BF^\alpha := \{ f^\alpha \mid f \in \BF\}$. 
We note that $\BF^\alpha \subset \BF$ for $0 \leq \alpha \leq 1$ because $f^\alpha$ is a composition of two 
Bernstein functions $x \rightarrow x^\alpha$ and $f$, see Theorem~\ref{thm:BFCM}. However $p \, f(p)$, $f \in \BF$, is 
in general not a Bernstein function, for example if $f(p) = 1 - \e^{-p}$. 

In \citet{HanWM2013} it was proved that $\mathcal{K}_\BF := \CBF \cap p^{1/2}\, \CBF$
is the set of all wave number functions compatible with the assumption that the relaxation modulus is LICM,
or, equivalently, that the creep compliance is a Bernstein function \citep{HanDuality}. 
The set of Bernstein functions $\BF$ is a subset of $\CrF$. We thus expect that the set 
$\mathcal{K}_\BF \subset \mathcal{K}_\CrF$. In order to check this fact directly 
we shall need the following theorem (Proposition~7.11 in \citet{BernsteinFunctions}):
\begin{theorem} \label{thm:auxCBF}
If $0 < \alpha \leq 1$ then 
$\CBF^\alpha = \{g \in \CBF \mid p^{1-\alpha}\, g(p) \in \CBF \}$.
\end{theorem}
\begin{corollary} \label{cor:aux}
$\CBF \cap p^{1/2}\, \CBF = p^{1/2} \, \CBF^{1/2}$.
\end{corollary}
\begin{proof}
Theorem~\ref{thm:logconvCBF} and Theorem~\ref{thm:sqrCBF} imply that $p^{1/2}\, \CBF^{1/2} \subset 
\CBF \cap p^{1/2}\, \CBF$.

If $f \in \CBF \cap p^{1/2}\, \CBF$ then $f(p) = p^{1/2}\, g(p)$ with $g \in \CBF$.
Theorem~\ref{thm:auxCBF} implies that $g \in \CBF^{1/2}$, hence $f \in p^{1/2}\, \CBF^{1/2}$,
q.e.d.
\end{proof} \hfill $\Box$

It follows from Corollary~\ref{cor:aux} that   
$\CBF \cap p^{1/2}\, \CBF = p^{1/2}\, \CBF^{1/2} \subset p^{1/2}\, \BF^{1/2}$, as expected. 

Returning to the case of a creep compliance in the CrF class, note that 
the wavenumber function has the form $\kappa(p) = p^{1/2} \, f(p)^{1/2}$, where $f$ has the 
integral representation 
\begin{equation}
f(p) = a + b \, p + \int_{]0,\infty[} \left[ 1 - \e^{-p r}\right] \,\nu(\dd r)
\end{equation}
with $a, b \geq 0$ and a non-negative Radon measure $\nu$ satisfying the inequalities \eqref{eq:twoineq}. 

Using Theorem~\ref{thm:lim} it is now easy to prove that 
$\kappa(p) = b^{1/2} \, p + R(p)$, where $R(p) = \oo[p]$ for $p \rightarrow \infty$. The phase function
in Green's functions has the form 
$$\exp(p\, t - \kappa(p) \vert x \vert) = \exp( p (t - b^{1/2}\, \vert x \vert) - R(p)\vert x
\vert)$$
with $p = -\ii \omega$, $\omega \in \mathbb{R}$. 
The arguments used in the proof of Theorem~7.2 in \citet{HanWM2013} can be used to demonstrate that the Green's function vanishes 
for $\vert x \vert > b^{-1/2}\, t$, hence $b^{-1/2}$ is the wavefront speed.

\section{Main theorem.}
\label{sec:proof}

Behavior of the attenuation and dispersion function for real frequencies (imaginary $p$) is of particular
interest. We shall therefore examine the behavior of the analytic continuation of a Bernstein function 
$f(p)$ on the imaginary axis. 

Analytic continuation of $f$ to the imaginary axis yields a function 
$F(\omega) := f(-\ii \omega)$, $\omega \in \mathbb{R}$. 
$f_\mathrm{R}(\omega) := \re F(-\ii \omega)$ and $f_\mathrm{I}(\omega) := \im F(-\ii \omega)$.
Hence
\begin{gather} \label{eq:au1}
f_\mathrm{R}(\omega) = a + \int_{]0,1]} [ 1 - \cos(r \omega)] \, \nu(\dd r) + 
\int_{]1,\infty[} [ 1 - \cos(r \omega)] \, \nu(\dd r)\\ \label{eq:au2}
f_\mathrm{I}(\omega) = \int_{]0,1]} \sin(r \omega) \, \nu(\dd r) + \int_{]1,\infty[} \sin(r \omega) \, \nu(\dd r)
\end{gather}

\begin{lemma} \label{lem:au}
The integrals in eqs.~(\ref{eq:au1}-\ref{eq:au2}) are bounded by linear functions of the variable $\omega$.
\end{lemma}
\begin{proof}  
Note that $\vert \sin(x)/x \vert \leq 1$ (the numerator and the denominator vanish at 0 and the absolute value of the derivative of the numerator does not exceed the value of the derivative of the numerator equal to 1). 
Hence $\vert \sin(r \omega)/r \vert \leq \vert \omega \vert$. In the same way one shows that $\vert (1 -  \cos(r \omega)/r \vert \leq \vert \omega \vert$. Hence   
$$\left\vert \int_{]0,1]} \sin(r \omega) \, \nu(\dd r) \right\vert = \left\vert \int_{]0,1]} \frac{\sin(r \omega)}{r} \, r \nu(\dd r) \right\vert \leq \vert \omega \vert \int_{]0,1]} r \, \nu(\dd r)$$
and 
$$\left\vert \int_{]0,1]} [ 1 - \cos(r \omega)]\, \nu(\dd r) \right\vert = \left\vert \int_{]0,1]} 
\frac{1 - \cos(r \omega)}{r} \, r \nu(\dd r) \right\vert \leq \vert \omega \vert \int_{]0,1]} r \nu(\dd r)$$
Furthermore
$$\left\vert \int_{]1,\infty[} \sin(r \omega) \, \nu(\dd r) \right\vert \leq \int_{]1,\infty[} \nu(\dd r) < \infty$$
$$\left\vert \int_{]1,\infty[} [ 1 - \cos(r \omega)] \, \nu(\dd r) \right\vert  \leq 2 \int_{]1,\infty[} \nu(\dd r) < \infty$$

Hence $\vert f_\mathrm{R}(\omega) \vert \leq a + C \vert \omega\vert$ and $\vert f_\mathrm{I}(\omega)\vert \leq b + \vert \leq  
D \vert \omega \vert$ for some positive reals $a, b, C, D$. 
\end{proof} 

\begin{theorem}
\begin{gather}
\vert \kappa_\mathrm{I}(\omega) \vert \leq K + L\,\vert\omega\vert
\\
\vert \kappa_\mathrm{R}(\omega) \vert \leq K + L \, \vert\omega\vert
\end{gather}
for some real constants $K$ and $L$.
\end{theorem}
\begin{proof}
If $m = \max\{a, C\}$, $M = \max\{b, D\}$
then $$\vert F(\omega)\vert = \sqrt{f_\mathrm{R}(\omega)^2 + f_\mathrm{I}(\omega)^2} \leq \sqrt{2} \, \left( m + M \, \vert \omega\vert\right)$$
and therefore
$$
\vert \kappa(-\ii \omega) \vert \leq 2^{1/4} \, \left(
m \, \vert \omega\vert + M\, \vert \omega\vert^2 \right)^{1/2} \leq 2^{1/4} \,
\left(N + M^{1/2} \, \vert \omega \vert \right)$$ 
where $N := m/(2 M^{1/2})$. The thesis follows with $K = 2^{1/4}\, N$ and $L = 2^{1/4}\, M^{1/2}$. 
\end{proof}

It follows that the attenuation is majorized by a linear function of frequency. In the high frequency 
range it increases at most at a linear rate. This bound can perhaps be sharpened. 

\section{Concluding remarks.}

It has been shown that the assumption that the creep compliance is non-decreasing and concave implies 
that the attenuation function increases at a sublinear rate in the high-frequency range. 
This result is somewhat weaker than for materials with completely monotonic relaxation moduli and Bernstein class creep 
compliances \citep{SerHan2010,HanWM2013}. 

This result shows that some acoustic wave equations applied to ultrasonic 
tests of materials and especially for ultrasound in bio-tissues still popular in literature
\citep{Szabo1,Szabo2,KellyMcGoughMeerschaert08} are not compatible with viscoelasticity. Such equations are
based on the assumption of a power law attenuation with an exponent exceeding 1. We have thus shown that
such equations are incompatible with the assumption that the material is linear viscoelastic and the creep
compliance is non-decreasing and concave. 

The results obtained in this paper are compatible with experimental data. In the frequency range covered by
measurements the frequency dependence of attenuation is represented by a power law with an exponent $1 \leq
 \alpha \leq 2$ over less than three decades of frequency expressed in MHz. This fact is consistent with 
the results obtained in \citet{HanWM2013} for low-frequency attenuation in viscoelastic solids. 
It should however be borne in mind that the sound absorption measured in heterogeneous media includes loss due to 
scattering which is here not accounted for.

\appendix

\section{}
\label{app:A}

\begin{definition}
A real function $f$ on $[0,\infty[$ is said to be a complete Bernstein function if there
is a Bernstein function $g$ such that $f(x) = x^2\, \tilde{g}(x)$. 
\end{definition}

\begin{equation} \label{eq:xyc}
f(x) = a + b\, x + x \int_{]0,\infty[} \frac{\mu(\dd r)}{x + r}
\end{equation}

\begin{theorem}  \label{thm:CBFanalytic} \citep{Jacob01I} \\
Every CBF can be continued to an analytic function $F$ on $\mathbb{C}_-$ satisfying the
following conditions:
\begin{enumerate}[(i)]
\item $F$ assumes real values on $\mathbb{R_+}$;
\item a finite limit $\lim_{\substack{x\rightarrow 0\\x > 0}}\, F(x)$ exists and 
is non-negative;
\item $\sgn(\im[F(z)]) = \sgn(\im[z])$; \label{it:CBFanaliticiii}
\item
The function 
$F$  has an unique integral representation
$$F(z) = a + b \, z + \int_{]0,\infty[\;} \frac{t z - 1}{t + z} \rho(\dd t)$$
with $a, b \geq 0$ and a positive Radon measure $\rho$ satisfying the inequality
satisfying the inequality
$$\int_{]0,\infty[\;} \frac{\rho(\dd t)}{1 + t} < \infty$$
\label{it:CBFanaliticiv}
\end{enumerate}

If a complex valued function $F$ satisfies either the conditions (i)--(iii) or
condition (iv) then the restriction of $F$ to $\overline{\mathbb{R}_+}$ is a CBF.
\end{theorem}

This theorem has the following corollary:
\begin{theorem} \label{thm:sqrCBF}
If $f$ is a CBF and $0 \leq \alpha \leq 1$ then $f(x)^\alpha$ is also a CBF.
\end{theorem}

\begin{theorem} \label{thm:logconvCBF}
If $f, g \in \CBF$ and $0 \leq \alpha \leq 1$ then the pointwise product 
$h := f^\alpha \, g^{1-\alpha} \in \CBF$. 
\end{theorem}
\begin{proof}
The functions $f, g$ map $\mathbb{R}_+$ into $\overline{\mathbb{R}_+}$,
hence $h$ has the same property.

The functions $f, g$ are holomorphic in $\mathbb{C}_\pm$ and 
map $\mathbb{C}_\pm$ into $\overline{\mathbb{C}_\pm}$,
hence $h(z)$ is holomorphic in $\mathbb{C}+$ and $0 \leq \arg h(z) 
\leq \alpha \, \upi + (1-\alpha)\, \upi = \upi$. By Theorem~\ref{thm:CBFanalytic}
$h \in \CBF$.
\end{proof}

\end{document}